\documentclass[11pt]{article}
\usepackage{a4wide}

\usepackage{complexity}
\usepackage{graphicx}
\usepackage{framed}
\usepackage{amsmath}
\usepackage{amsthm}
\usepackage{amsfonts}
\usepackage{mathtools}
\usepackage{cite}
\usepackage{xspace}

\usepackage{amsmath}
\usepackage{amssymb}
\usepackage{graphicx}
\usepackage{url}
\usepackage{comment}
\usepackage{enumitem}

\usepackage{amsfonts}
\usepackage{amsthm}
\usepackage{comment}
\usepackage{thmtools}
\usepackage{thm-restate}
\usepackage{mdframed}
\usepackage{booktabs}

\newtheorem{theorem}{Theorem}[]
\newtheorem{lemma}[theorem]{Lemma}
\newtheorem{corollary}[theorem]{Corollary}

\newtheorem{proposition}[theorem]{Proposition}

\usepackage[noend,linesnumbered,ruled,vlined]{algorithm2e}
\DontPrintSemicolon
\SetArgSty{textup}

\usepackage{xcolor}
\usepackage{colortbl}
\colorlet{tableheadcolor}{gray!25} 
\colorlet{tablerowcolor}{gray!20} 
\usepackage{graphicx}
\usepackage{booktabs}
\usepackage{pdflscape}
\usepackage{mdframed}
\usepackage{subfig}
\usepackage{mathtools}

\DeclareMathOperator{\tw}{tw}

\renewcommand{\P}{\textsc{P}}
\renewcommand{\NP}{\textsc{NP}}

\usepackage{xspace}
\usepackage{framed}

\newcommand{\ProblemFormat}[1]{{\sc #1}}
\newcommand{\ProblemName}[1]{\ProblemFormat{#1}\xspace}

\DeclareMathOperator{\adj}{E}
\DeclareMathOperator{\colnei}{NeiCols}

\newcommand{\SB}{\{\,} 
\newcommand{\SM}{\;{|}\;} 
\renewcommand{\SE}{\,\}}

\newcommand{\probMonotoneNaeThreeSAT}{\ProblemName{Monotone Not-All-Equal $3$-SAT}}

\newcommand{\probkCol}{\ProblemName{$k$-Coloring}}

\newcommand{\probPlanarThreeCol}{\ProblemName{Planar $3$-Coloring}}

\newcommand{\probtEDkCOL}{\ProblemName{$d$-Exact Defective $k$-Coloring}}
\newcommand{\probTwoEDkCOL}{\ProblemName{$2$-Exact Defective $k$-Coloring}}
\newcommand{\probTwoEDTwoCOL}{\ProblemName{$2$-Exact Defective $2$-Coloring}}
\newcommand{\probOneEDTwoCOL}{\ProblemName{$1$-Exact Defective $2$-Coloring}}
\newcommand{\probtEDTwoCOL}{\ProblemName{$d$-Exact Defective $2$-Coloring}}
\newcommand{\probtEDThreeCOL}{\ProblemName{$d$-Exact Defective $3$-Coloring}}
\newcommand{\probtplustwoEDTwoCOL}{\ProblemName{$(d+2)$-Exact Defective $2$-Coloring}}

\newcommand{\probCompleteCover}{\ProblemName{$K_r$-factor}}

\SetKwInput{KwInput}{Input}
\SetKwInput{KwOutput}{Output}
\SetKw{Continue}{continue}

\makeatletter
\g@addto@macro\@floatboxreset\centering
\makeatother

\newcount\bsubfloatcount
\newtoks\bsubfloattoks
\newdimen\bsubfloatht

\makeatletter
\newcommand{\bsubfloat}[2][]{%
  \sbox\z@{#2}%
  \ifdim\bsubfloatht<\ht\z@
    \bsubfloatht=\ht\z@
  \fi
  \advance\bsubfloatcount\@ne
  \@namedef{bsubfloat\romannumeral\bsubfloatcount}{%
    \subfloat[#1]{\vbox to\bsubfloatht{\hbox{#2}\vfill}}}%
}
\newcommand{\resetbsubfloat}{\bsubfloatcount\z@\bsubfloatht=\z@}
\makeatother

\title{Exact defective colorings of graphs}

\author{James Cumberbatch\thanks{Department of Mathematics, Purdue University, USA} \and Juho Lauri\thanks{Helsinki, Finland} \and Christodoulos Mitillos\thanks{Department of Mathematics and Statistics, University of Cyprus, Cyprus}}

\begin{document}
\maketitle

\begin{abstract}
An \emph{exact $(k,d)$-coloring} of a graph $G$ is a coloring of its vertices with $k$ colors such that each vertex $v$ is adjacent to exactly $d$ vertices having the same color as $v$.
The \emph{exact $d$-defective chromatic number}, denoted $\chi_d^=(G)$, is the minimum $k$ such that there exists an exact $(k,d)$-coloring of $G$.
In an exact $(k,d)$-coloring, which for $d=0$ corresponds to a proper coloring, each color class induces a $d$-regular subgraph.
We give basic properties for the parameter and determine its exact value for cycles, trees, and complete graphs. 
In addition, we establish bounds on $\chi_d^=(G)$ for all relevant values of $d$ when $G$ is planar, chordal, or has bounded treewidth.
We also give polynomial-time algorithms for finding certain types of exact $(k,d)$-colorings in cactus graphs and block graphs.
Our main result is on the computational complexity of \probtEDkCOL in which we are given a graph $G$ and asked to decide whether $\chi_d^=(G) \leq k$.
Specifically, we prove that the problem is $\NP$-complete for all $d \geq 1$ and $k \geq 2$.
\end{abstract}

\section{Introduction}
Let $G=(V,E)$ be an undirected (simple) graph. 
An \emph{exact $(k,d)$-coloring} of $G$ is a coloring of its vertices such that each vertex $v$ has exactly $d$ neighbors having the same color as $v$.
We define \emph{the exact $d$-defective chromatic number}, denoted $\chi_d^=(G)$, as the minimum $k$ such that there exists an exact $(k,d)$-coloring of~$G$.
It is not difficult to see that for particular values of $k$ and $d$, it can be that a graph $G$ does not admit an exact $(k,d)$-coloring.
For example, when $G$ is an odd cycle, an exact $(k,1)$-coloring does not exist for any positive $k$.
In such cases, we write $\chi^=_d(G) = \infty$.
Clearly, an exact $(k,0)$-coloring is a proper $k$-coloring, and thus $\chi_0^=(G) = \chi(G)$, where $\chi(G)$ is the chromatic number of $G$.
Further, our definitions are closely related to a \emph{$(k,d)$-coloring} in which one requires \emph{at most} $d$ neighbors of the same color as the vertex $v$. 
This, in turn, gives rise to the \emph{$d$-defective chromatic number} of a graph $G$, denoted $\chi_d(G)$, which is the minimum $k$ such that there exists a $(k,d)$-coloring of~$G$.
As such, it is also clear more generally that $\chi_0^=(G) = \chi_0(G) = \chi(G)$.
Finally, one can readily observe that in the case of an exact $(k,d)$-coloring, the color classes induce $d$-regular subgraphs, whereas in the case of a $(k,d)$-coloring, the color classes induce subgraphs of maximum degree $d$.

Defective colorings (i.e., $(k,d)$-colorings) were introduced nearly simultaneously by Burr and Jacobson~\cite{Andrews1985}, Harary and Jones~\cite{Harary1985}, and Cowen~{et al.}~\cite{Cowen1986}.
As a main result, the latter set of authors gave a complete characterization of all $k$ and $d$ such that every planar or outerplanar graph is $(k,d)$-colorable.
For general graphs, Cowen~{et al.}~\cite{Cowen1997} showed that by applying a result of Lov{\'a}sz~\cite{Lovasz1966} a greedy procedure establishes that $\chi_d(G) \leq \lfloor \Delta / (d+1) \rfloor + 1$ for any $d$.
Further, the authors proved that for all $k \geq 3$ and $d \geq 2$, it is $\NP$-complete to decide whether a given graph admits a $(k,d)$-coloring.
In addition, they showed that deciding whether a given planar graph has a $(3,1)$-coloring remains $\NP$-complete, and also that for each $d \geq 1$ it is $\NP$-complete to decide if a planar graph is $(2,d)$-colorable.
For general results, we refer the interested reader to the survey of Frick~\cite{Frick1993}.
For other algorithmic results on defective coloring, see Belmonte~{et al.}~\cite{Belmonte2017, Belmonte2020}.

While not under the name we introduce here, exact $(k,d)$-colorings have been studied before but considerably less heavily than $(k,d)$-colorings.
In particular, in his seminal work from 1978, Schaefer~\cite{Schaefer1978} proved that it is $\NP$-complete to decide whether the vertices of a given planar cubic graph can be 2-colored such that each vertex has exactly one neighbor of the same color as itself.
Equivalently, this result states that it is $\NP$-complete to decide whether a given planar cubic graph admits an exact $(2,1)$-coloring.
Finally, Yuan and Wang~\cite{Yuan2003} studied $\chi_1^=(G)$ under the name of induced matching partition number, gave an upper bound of $2 \Delta - 1$ for it, and characterized the graphs achieving this bound.
Besides these, we are not aware of any other results on exact $(k,d)$-colorings.
Furthermore, Gera~{et al.}~\cite{Gera2018} drew attention to the problem by writing that it is one of six vertex-partitioning problems that has ``received little or no attention''.

We find it interesting to contrast the result of Cowen~{et al.}~\cite{Cowen1997}, that deciding whether a graph admits a $(2,1)$-coloring is $\NP$-complete, with the usual notion of proper 2-coloring.
As is well-known, proper 2-colorability can be determined with a simple graph search and is thus solvable in polynomial time.
In other words, from a complexity-theoretic perspective, partitioning the vertices of a graph into independent sets can be considerably easier than partitioning them into subgraphs of bounded maximum degree. 
As such, it is natural to ask: how different is it to consider a partition into $d$-regular subgraphs?

\paragraph{Our results}
We initiate the systematic study of exact defective colorings.
After stating key definitions and existing results in Section~\ref{sec:prelims}, we continue to give the following results.

\begin{itemize}
\item In Section~\ref{sec:basic}, we look at some basic results of exact defective colorings. These include characterizations of such colorings for cycles, wheel graphs, trees, and complete graphs.

\item In Section~\ref{sec:general-results}, we consider how exact defective colorings relate to certain structural properties of graphs. We establish results on the presence of exact defective colorings for induced minor closed graphs, including planar graphs, outerplanar graphs, graphs of bounded treewidth, and chordal graphs. In addition, we prove that the exact $d$-defective chromatic number is incomparable with the chromatic number of a graph. 

\item In Section~\ref{sec:complexity}, we establish complexity results for \probtEDkCOL (see Section~\ref{sec:prelims} for a definition) for all parameter values of $k$ and $d$. Specifically, we show that this problem is $\NP$-complete for all $d \geq 1$ when $k \geq 2$.
Observing that for $k = 1$ the problem is trivially polynomial to solve, we fully characterise the complexity of this family of problems.
Additionally, for planar graphs, we show that \probtEDThreeCOL remains $\NP$-complete for every valid value of $d$.

\item In Section~\ref{sec:algos}, we show that \probtEDkCOL is \emph{fixed-parameter tractable}\footnote{A problem is said to be fixed-parameter tractable with respect to a parameter $k$ if it can be solved in time $f(k) \cdot n^{O(1)}$, where $f$ is some computable function depending only on $k$ and $n$ the input size.} (FPT) when parameterized by treewidth. We then give a combinatorial polynomial-time algorithm for solving \probTwoEDkCOL for cactus graphs. Finally, we show that a linear-time algorithm exists for solving \probtEDkCOL on block graphs.
\end{itemize}

\section{Preliminaries}
\label{sec:prelims}
For a positive integer $n$, we write $[n] = \{1,2,\ldots,n\}$.

In this section, we define the graph-theoretic concepts most central to our work.
For graph-theoretic notation not defined here, we refer the reader to~\cite{Diestel2010}.
We also briefly introduce decision problems our hardness results depend on.

\paragraph{Graph parameters and classes} 
All graphs we consider are undirected and simple.
For a graph $G$, we denote by $V(G)$ and $E(G)$ its vertex set and edge set, respectively.
To reduce clutter, an edge $\{u,v\}$ is often denoted as $uv$. 
Two vertices $x$ and $y$ are \emph{adjacent} (or \emph{neighbors}) if $xy$ is an edge of $G$. The \emph{neighborhood} of a vertex $v$, denoted by $N(v)$, is the set of all vertices adjacent to $v$. 

A \emph{vertex-coloring} (or simply \emph{coloring}) is a function $c : V \to [k]$ assigning a color from $[k]$ to each vertex of a graph $G=(V,E)$. 
The coloring is said to be \emph{proper} if $c(u) \neq c(v)$ for every $uv \in E$. 
A graph $G$ is said to be \emph{$k$-colorable} if there exists a proper vertex-coloring using $k$ colors for it. 
The minimum $k$ for which a graph $G$ is $k$-colorable is known as its \emph{chromatic number}, denoted by $\chi(G)$.
In particular, a 2-colorable graph is \emph{bipartite}.

A graph is \emph{planar} if it can be embedded in the plane with no crossing edges. A graph is \emph{outerplanar} if it has a crossing-free embedding in the plane such that all vertices are on the same face. 
In a \emph{cactus graph}, every maximal biconnected component, known as a \emph{block}, is a cycle or a $K_2$.
Cactus graphs form a subclass of outerplanar graphs.

Finally, we mention the following well-known structural measure for ``tree-likeness'' of graphs.
A \emph{tree decomposition} of $G$ is a pair 
$(T,\{X_i : i\in I\})$
where $X_i \subseteq V$, $i\in I$, and $T$ is a tree with elements
of $I$ as nodes
such that:
\begin{enumerate}
\item for each edge $uv\in E$, there is an $i\in I$ such that $\{u,v\} 
\subseteq X_i$, and
\item for each vertex $v\in V$, $T[\SB i\in I \SM v\in X_i \SE]$ is a tree with at least one node.
\end{enumerate}
The \emph{width} of a tree decomposition is $\max_{i \in I} |X_i|-1$.
The \emph{treewidth} of $G$, denoted by $\tw(G)$, is the minimum width taken over all tree decompositions of $G$. 
It can be noted that outerplanar graphs, and consequently cactus graphs, have treewidth at most two.

\paragraph{Decision problems}
For completeness, we define some of the computational problems relevant to our results here.
In \probkCol, we are given a graph $G$ and the goal is to decide whether or not $G$ admits a $k$-coloring.
This problem is well-known to be $\NP$-complete for every $k \geq 3$.
Further, for $k = 3$, the problem remains $\NP$-complete when restricted to 4-regular planar graphs~\cite{Dailey1980}.

Our main focus is on \probtEDkCOL, where we are given a graph $G$ and the goal is to decide whether $\chi_d^=(G) \leq k$, i.e., whether the vertices of $G$ can be colored in $k$ colors such that each vertex $v$ of $G$ has exactly $d$ neighbors having the same color as $v$.
In other words, the goal is to decide whether the vertices of $G$ can be partitioned into at most $k$ color classes such that each class induces a $d$-regular subgraph.

\paragraph{Monadic second order logic}
Let us denote individual variables by lowercase letters $x$, $y$, $z$ and set variables by uppercase letters $X$, $Y$, $Z$. 
\emph{Formulas} of MSO$_2$ logic are constructed from atomic
formulas $I(x,y)$, $x\in X$, and $x = y$ using the connectives $\neg$
(negation), $\wedge$ (conjunction) and existential quantification
$\exists x$ over individual variables or existential
quantification $\exists X$ over set variables. Individual variables
range over vertices and edges, and set variables range either over sets of
vertices or over sets of edges. The atomic formula $I(x,y)$ expresses that vertex $x$ is incident to edge $y$, $x = y$
expresses equality, and $x\in X$ expresses that $x$ is in the set
$X$. From this, we define the semantics of MSO$_2$ logic
in the standard way.

MSO$_1$ logic is defined similarly as MSO$_2$ logic, with the following distinctions. Individual variables range only over vertices, and set variables only range over sets of vertices. The atomic formula $I(x,y)$ is replaced by $E(x,y)$, which expresses that vertex $x$ is adjacent to vertex $y$.
In other words, MSO$_1$ is a weaker logic which forbids quantification over edge subsets.

\emph{Free and bound variables} of a formula are defined in the usual way. A
\emph{sentence} is a formula without free variables. 
It is well-known that MSO$_2$ formulas can be checked efficiently on graphs of bounded treewidth.

\begin{theorem}[Courcelle~\cite{Courcelle90a}]
\label{thm:msotreewidth}
Let $\phi$ be a fixed \emph{MSO}$_2$ sentence and $p$ be a positive constant. Given an $n$-vertex graph $G$ of treewidth at most $p$, it is possible to decide whether $G\models \phi$ in time $O(n)$.
\end{theorem}

In a similar spirit, MSO$_1$ formulas can be checked efficiently on graphs of bounded \emph{cliquewidth}~\cite{CourcelleMakowskyRotics00} (or, equivalently, \emph{rankwidth}~\cite{GanianHlineny10}). 
Specifically, while the formula can be checked in linear time if a suitable rank- or clique-decomposition is given, such a decomposition itself can be found in cubic time.

\begin{theorem}[Courcelle, Makowsky, and Rotics~\cite{CourcelleMakowskyRotics00}, Ganian and Hlin\v{e}n{\'y}~\cite{GanianHlineny10}]
\label{fact:msorankwidth}
Let $\phi$ be a fixed \emph{MSO}$_1$ sentence and $p$ be a positive constant. Given an $n$-vertex graph $G$ of cliquewidth at most $p$, it is possible to decide whether $G\models \phi$ in time $O(n^3)$.
\end{theorem}

\section{Basic properties}
\label{sec:basic}
In an exact $(k,d)$-coloring, each color class induces a $d$-regular graph. Thus, the following is an easy observation we will tacitly use later on.
\begin{proposition}
A graph $G$ does not admit an exact $(k,d)$-coloring for any $d > \delta(G)$.
\end{proposition}

Let $M = \{x_1y_1, x_2y_2, \ldots, x_ky_k \}$ be a matching of a graph $G$.
We denote by $G / M$ the (simple) graph obtained from $G$ by contracting $M$, i.e., we replace each pair of vertices $x_i y_i$ for $1 \leq i \leq k$ by a new vertex $z_i$ whose neighbors are all the neighbors of $x_i$ and $y_i$.
Yuan and Wang~\cite{Yuan2003} gave the following results.
\begin{lemma}[Yuan and Wang~\cite{Yuan2003}]
\label{lem:chi1-perf}
A graph $G$ with at least one perfect matching has 
\begin{equation*}
\chi^=_1(G) = \min \{ \chi(G / M) \mid M \text{ is a perfect matching of } G\}.
\end{equation*}
\end{lemma}

\begin{theorem}[Yuan and Wang~\cite{Yuan2003}]
\label{thm:chi1-2d1-bound}
A connected graph $G$ having at least one perfect matching has $\chi^=_1(G) \leq 2 \Delta(G) - 1$. Further, equality holds if and only if $G$ is isomorphic to either $K_2$, $C_{4k+2}$, or the Petersen graph.
\end{theorem}
\noindent With these results at hand, let us consider exact $(k,d)$-colorings of some structured graphs.
\begin{proposition}
For every $n \geq 3$, it holds that
\begin{equation*}
\chi^=_1(C_n) = \begin{cases}
2 &\text{if $n$ is divisible by four,}\\
3 &\text{if $n$ is even but not divisible by four, and}\\
\infty &\text{otherwise}.
\end{cases}
\end{equation*}
\end{proposition}
\begin{proof}
By definition, in any exact $(k,1)$-coloring each vertex $v$ of $C_n$ must be adjacent to precisely one neighbor with the same color as itself.
Now, as any $v$ has degree two, we must have $k > 1$ for otherwise $v$ would be forced to neighbor two vertices sharing the same color as itself. 

Suppose that $n$ is a multiple of four.
We order the edges of $C_n$ clockwise and color both the endpoints of the even ones alternately with colors 1 and 2. 
This is an exact $(2,1)$-coloring proving that $\chi^=_1(C_n) = 2$.
On the other hand, suppose that $n$ is even but not a multiple of four.
In this case we have by Theorem~\ref{thm:chi1-2d1-bound} that $\chi^=_1(C_n) = 2 \Delta - 1 = 3$.
Finally, when $n$ is odd, $C_n$ has no perfect matching and thus no solution exists.
\end{proof}

\begin{proposition}
For every $n \geq 2$, it holds that $\chi^=_2(C_n) = 1$.
\end{proposition}
\begin{proof}
We meet the trivial lower bound by coloring the cycle monochromatic.
\end{proof}

A \emph{wheel graph} on $n$ vertices, denoted as $W_n$, is obtained by adding a universal vertex to a cycle $C_{n-1}$.
\begin{proposition}
For every $n \geq 4$, it holds that
\begin{equation*}
\chi^=_1(W_n) = \begin{cases}
2 &\text{if $n = 4$,}\\
3 &\text{if $n > 4$ is even, and}\\
\infty &\text{otherwise}.
\end{cases}
\end{equation*}
\end{proposition}
\begin{proof}
When $n$ is odd, $W_n$ has no perfect matching and thus no solution exists.
It is simple to check that $\chi^=_1(W_4) = 2$, so suppose that $n > 4$.
Consider any perfect matching $M$ of $W_n$.
It is not difficult to see that $W_n / M$ is isomorphic to a fan graph, that is, a path with a universal vertex added which has chromatic number three.
Therefore, the proof follows via Lemma~\ref{lem:chi1-perf}.
\end{proof}

Let us then consider trees. Note that every non-trivial tree has a vertex of degree one, so we need not consider exact $(k,d)$-colorings for any $d > 1$.

\begin{proposition}
A tree $T$ on at least three vertices has $\chi^=_1(T) = 2$ if and only if $T$ has a perfect matching. Furthermore, if $\chi^=_1(T) \neq 2$ then $\chi^=_1(T) = \infty$.
\end{proposition}
\begin{proof}
If $T$ has a perfect matching $M$, then $T/M$ is cycle-free (as it is also a tree) and thus 2-colorable. Therefore, by Lemma~\ref{lem:chi1-perf}, we have that $\chi^=_1(T) = 2$. Since every exact $(k, 1)$-coloring for any graph defines a perfect matching, the result follows.
\end{proof}

We will now give a simple lower bound for any exact $d$-defective chromatic number based on the structure of the given graph. This is a more general result to the observation that the chromatic number of a graph (i.e., when $d = 0$) is not less than the \emph{clique number} of the graph, that is, the order of its largest complete subgraph.

\begin{proposition}
For every $d \geq 0$, any graph $G$ has $\chi^=_d(G) \geq \lceil \omega / (d+1) \rceil$, where $\omega$ is the clique number of $G$.
\end{proposition}
\begin{proof}
Suppose this was not the case, i.e., that there is an exact $(k,d)$-coloring of $G$ using fewer than $\lceil \omega / (d+1) \rceil$ colors.
As all vertices of $G$ are colored, there is a color which appears more than $d+1$ times in a clique of size $\omega$, contradicting the fact that we have an exact $(k,d)$-coloring.
\end{proof}
\noindent By applying the result, we find the exact $d$-defective chromatic number for complete graphs.
\begin{corollary}
For every $d \geq 0$ and $n \geq 1$, 
\begin{equation*}
\chi^=_d(K_n) = \begin{cases}
n / (d+1) &\text{if $n$ is divisible by $(d+1)$, and}\\
\infty &\text{otherwise}.
\end{cases}
\end{equation*}
\end{corollary}

\section{Exact $d$-defective coloring and proper coloring}
\label{sec:general-results}

Let $G=(V,E)$ be a graph.
For any positive integer $d$, let $\mathcal{R}_d(G)$ be the collection of all possible partitions of $V$ into $d$-regular induced subgraphs.
For instance, when $d = 2$, an element of $\mathcal{R}_d(G)$ is a partition of $G$ into induced cycles each of which has length at least three. 
In other words, we arrive at the following generalization of Lemma~\ref{lem:chi1-perf}.
\begin{lemma}
\label{lem:chid-perf}
Let $d$ be a positive integer. A graph $G$ with non-empty $\mathcal{R}_d(G)$ has
\begin{equation*}
\chi^=_d(G) = \min \{ \chi(G / H) \mid H \in \mathcal{R}_d(G) \}.
\end{equation*}
\end{lemma}
\begin{proof}
By definition, the color classes of any exact $(k,d)$-coloring of $G$ give a partition of $G$ into induced $d$-regular subgraphs, corresponding to one element in $\mathcal{R}_d(G)$.
Conversely, any element of $H \in \mathcal{R}_d(G)$ gives a valid exact $(k,d)$-coloring for $G$.
Indeed, we can blow-up $G / H$ to obtain $G$ and color the set of vertices, say $G_H$, corresponding to an element of $H$ with the same color.
As $G/H$ was properly colored, we are guaranteed the vertices in $G_H$ do not see their color on adjacent vertices outside of $G_H$.
\end{proof}
Before proceeding, recall that a graph $H$ is a \emph{minor} of $G$ if $H$ can be obtained from $G$ by deleting edges and vertices and by contracting edges.
More restrictively, $H$ is an \emph{induced minor} of $G$ if $H$ can be obtained from $G$ by only contracting edges. 
Now, consider any graph $G$ that belongs to an induced minor closed family of graphs $\mathcal{G}$. We can perform an arbitrary sequence of edge contractions on $G$ to obtain $G'$, and by definition $G'$ is also a member of $\mathcal{G}$.
It is possible that $\chi(G') > \chi(G)$, but assuming we have a good (e.g., constant) bound on the chromatic number of any member of $\mathcal{G}$, we can still observe the following.
\begin{theorem}
\label{thm:general-ted}
Let $d$ be a positive integer and $\mathcal{G}$ a class of induced minor closed graphs. Every graph $G \in \mathcal{G}$ has $\chi_d^=(G) \leq \max_{H \in \mathcal{G}} \chi(H)$ if and only if $\mathcal{R}_d(G)$ is non-empty.
\end{theorem}

\begin{corollary}
\label{cor:explicit-ted-bounds}
Let $d$ be a positive integer. A graph $G$ has 
\begin{enumerate}[label={(\roman*)}]
\item $\chi_d^=(G) \leq 4$ when $G$ is planar for $1 \leq d \leq 5$,
\item $\chi_d^=(G) \leq 3$ when $G$ is outerplanar for $1 \leq d \leq 2$, 
\item $\chi_d^=(G) \leq p + 1$ when $G$ has treewidth at most $p$ for $1 \leq d \leq p$, and
\item $\chi_d^=(G) \leq \omega(G)$ when $G$ is chordal for $1 \leq d \leq \omega(G)$
\end{enumerate}
if and only if $\mathcal{R}_d(G)$ is non-empty. Moreover, in each case, $\chi_d^=(G) = \infty$ when $\mathcal{R}_d(G)$ is empty or $d$ is not in the specified range.
\end{corollary}
\begin{proof}
For each case, we apply Theorem~\ref{thm:general-ted} combined with a known upper bound on $\chi(G)$.
In addition, we apply elementary results to bound the values of $d$.

To obtain (i), observe that $\chi(G) \leq 4$ by the four color theorem. Further, as each planar graph has a vertex of degree at most~5, we need not consider $d > 5$.
Similarly, for (ii), it is known that a simple outerplanar graph has $\chi(G) \leq 3$ as shown by Proskurowski and Sys{\l}o~\cite{Proskurowski1986}.
Each outerplanar graph has a vertex of degree at most~2, so we need not consider $d > 2$.
Thirdly, (iii) is proved by the fact that each graph of treewidth at most $p$ has $\chi(G) \leq p + 1$.
Moreover, each graph of treewidth at most $p$ has a vertex of degree at most $p$, so we need not consider $d > p$.
Finally, (iv) follows as the clique number of an induced minor $H$ of $G$ is at most $\omega(G)$ combined with the well-known fact that any chordal graph $G$ has $\chi(G) = \omega(G)$.
\end{proof}
\noindent While Theorem~\ref{thm:general-ted} does not give us insight into whether $\mathcal{R}_d(G)$ is non-empty, it tells us that if it is, then a solution with a small number of colors is guaranteed to exist for some structured graph classes.
In fact, we will later on show that the problem of deciding whether $\mathcal{R}_d(G)$ is non-empty is $\NP$-complete, even when $G$ is planar (see Corollary~\ref{cor:tedcol-planar-npc}).

At this point, given how similar the bounds on exact $d$-defective chromatic number and the chromatic number seem to be, it is natural to ask whether either parameter can be used to bound the other. While the two parameters coincide when $d = 0$, it can easily be seen that for $d = 2$, $\chi_2^=(C_n) = 1 < \chi(C_n)$.
However, in what is to follow, we will show that the two parameters are incomparable for all values of $d > 0$.

Let us first demonstrate that the chromatic number of a graph can be arbitrarily larger than its exact $d$-defective chromatic number, for any $d > 0$.
For this, we use the \emph{Cartesian product} of graphs $G$ and $H$, denoted as $G \Box H$. As a reminder, $V(G \Box H) = V(G) \times V(H)$ and $(u_1,v_1)(u_2,v_2) \in E(G \Box H)$ when either $u_1u_2 \in E(G)$ and $v_1 = v_2$ or $u_1 = u_2$ and $v_1v_2 \in E(H)$.

\begin{lemma}
Let $d$ and $r$ be positive integers with $r > 1$ and let $G = K_2 \Box K_{(d + 1)r}$. Then $\chi(G) = (d + 1)r$ and $\chi_d^=(G) \leq r$.
\end{lemma}
\begin{proof}
Since $(d + 1)r \geq 4$, it is well-known that $\chi(G) = (d + 1)r$. Therefore, we need only show that $\mathcal{R}_d(G)$ contains at least one partition $H$ of $G$ into $d$-regular graphs so that $G / H$ is $r$-colorable. We can partition $K_{(d + 1)r}$ into $r$ copies of $K_{d + 1}$, which are $d$-regular graphs. We select one such partition and apply it to both copies of $K_{(d + 1)r}$ in $G$. Letting this be our $H$, we have $G / H = K_2 \Box K_r$. But since $r > 1$, $\chi(G / H) = r$, as required.
\end{proof}

From the above, we see that for every positive $d$ there is an infinite family of graphs for which the chromatic number and the exact $d$-defective chromatic number differ by a factor of $d + 1$. Moreover, since this works for arbitrarily large $r$, we can make the difference between the two factors as large as we like.

For the opposite direction, we will look at the \emph{categorical product} of graphs, $G \times H$. As before, $V(G \times H) = V(G) \times V(H)$, however $(u_1,v_1)(u_2,v_2) \in E(G \times H)$ when $u_1u_2 \in E(G)$ and $v_1v_2 \in E(H)$.

\begin{lemma}
Let $d$ and $r$ be positive integers with $r > 1$ and let $G = K_2 \times K_{(d + 1)r}$. Then $\chi(G) = 2$ and $\chi_d^=(G) = r$.
\end{lemma}
\begin{proof}
Similarly to the previous lemma, it is well-known that $\chi(G) = 2$. This time, we must show that any partition of $G$ into $d$-regular graphs will necessarily create at least $r$ partite sets. We observe that any $d$-regular induced graph in $G$ will either be a copy of $K_{d, d}$ or a copy of $K_{d + 1, d + 1} - M$, where $M$ is a perfect matching.
First, we will assume that $d > 1$. We claim that in this case any two of the above types of subgraphs of $G$ cannot belong to the same color class in an exact $2$-defective coloring. This can be easily seen by the fact that any two such subgraphs have to contain at least $d$ vertices from each of the two ``rows" of $G$, corresponding to the vertices of its $K_2$ factor. As such, there will always be some vertex of one subgraph adjacent to some vertex of the other, since $d > 1$.
Now, for the case where $d = 1$, it is easy to see that, if we have some copy of $K_{1, 1}$, it can potentially be combined into a color class with a unique corresponding copy of the same graph. However, the result of this will be a copy of $K_{2, 2} - M$. These subgraphs also cannot be in the same color class as any others, from the same reasoning as in the previous case.
This implies that in such a partition, each copy of $K_{d + 1, d + 1} - M$ will be its own color class, as will any copy of $K_{d, d}$ for $d > 1$. Since each such subgraph will have at most $2(d + 1)$ vertices and $G$ has $2r(d + 1)$ vertices, there must be at least $r$ color classes. Moreover, a partition into exactly $r$ copies of $K_{d + 1, d + 1} - M$ can always be found, since $r$ divides $r(d + 1)$, completing the proof.
\end{proof}

\noindent Combining the two results, we have the following theorem.
\begin{theorem}
Let $c$ and $d$ be positive integers. There exist graphs $G_1$ and $G_2$, such that $\chi_d^=(G_1) - \chi(G_1) > c$ and $\chi(G_2) - \chi_d^=(G_2) > c$.
\end{theorem}

\section{Complexity results}
\label{sec:complexity}
In this section, we look at the complexity of determining the viability of exact defective colorings. Since \probkCol is $\NP$-complete for $k \geq 3$, we will use a reduction from this problem to establish the complexity of all instances of \probtEDkCOL where $k \geq 3$.

\begin{lemma}
\label{lem:col-to-ed}
For every $k \geq 3$, \probkCol reduces in polynomial time to \probtEDkCOL for any $d \geq 1$.
\end{lemma}
\begin{proof}
Let $G$ be an instance of \probkCol for any $k \geq 3$.
In polynomial time, we will create the following instance $G'$ of \probtEDkCOL for any $d \geq 1$.

Let $H_{v,d}$ be a $d$-regular graph of constant size (for instance, let $H_{v,t} \simeq K_{d+1}$).
Construct $G'$ from $G$ by identifying each $v \in V(G)$ with an arbitrary vertex of a copy of $H_{v,d}$.

For the first direction, suppose that a proper vertex-coloring $c : V \to [k]$ witnesses that $G$ is $k$-colorable.
For each $v \in V(G)$, extend $c$ by coloring each uncolored vertex in $H_{v,d}$ by $c(v)$.
As $c$ is a proper vertex-coloring, $v$ has no neighbors outside of $H_{v,d}$ colored $c(v)$.
Further, as $H_{v,d}$ is $d$-regular, it holds that $v$ has exactly $d$ neighbors colored $c(v)$. Similarly, the newly added vertices in the copies of $H_{v,d}$ must also each have $d$ neighbors of their own color.
It follows that $G'$ has an exact $(k,d)$-coloring.

For the other direction, suppose that $G'$ admits an exact $(k,d)$-coloring $c' : V(G') \to [k]$.
Since the vertices of each $H_{v,d}$ other than $v$ must have $d$ neighbours of their color and $H_{v,d}$ is $d$-regular, each $H_{v,d}$ is monochromatic.
But then each $v$ has $d$ neighbors of its own color in $H_{v,d}$ and that color does not appear on the vertices also in $G$.
Thus, $c'$ restricted to $G$ is a proper $k$-coloring of $G$ as required.
This completes the proof.
\end{proof}

This also allows us to establish a related result, where further restrictions are placed on the input graphs.

\begin{corollary}
\label{cor:tedcol-planar-npc}
For every $1 \leq d \leq 5$, \probtEDThreeCOL remains $\NP$-complete when $G$ is planar and has maximum degree $d+4$.
\end{corollary}
\begin{proof}
Apply Lemma~\ref{lem:col-to-ed} but instead reduce from \probPlanarThreeCol, where the input graph is also 4-regular.
This problem is $\NP$-complete as shown by Dailey~\cite{Dailey1980}.
Finally, note that a planar $d$-regular graph $H_{v,d}$ exists for the given values of $d$.
Specifically, one can note that the degree of each vertex is either $d$ or $d+4$, so the proof follows.
\end{proof}
\noindent In particular, this result explains why Theorem~\ref{thm:general-ted} is the best we can hope for in terms of identifying $\mathcal{R}_d(G)$, in the sense that there is no polynomial time algorithm for deciding whether $\mathcal{R}_d(G)$ is non-empty, unless $\P = \NP$.

It can be easily seen that the reduction used in Lemma~\ref{lem:col-to-ed} will not work for the cases where $k = 2$. However, since the complexity of one such case ($d = 1$) is already known, we can use an inductive argument to cover others.

\begin{theorem}[Schaefer~\cite{Schaefer1978}]
\label{thm:1ed2col-npc}
\probOneEDTwoCOL is $\NP$-complete.
\end{theorem}

We can use this as a starting point, to prove $\NP$-completeness for a number of subproblems of this family.

\begin{lemma}
\label{ed2col-induction}
If \probtEDTwoCOL is $\NP$-complete, then \probtplustwoEDTwoCOL is  $\NP$-complete.
\end{lemma}
\begin{proof}
Given an instance $G$ of \probtEDTwoCOL, we construct an instance $G'$ of \probtplustwoEDTwoCOL in polynomial-time as follows. 
For each vertex $v_i$ in $G$, we create a $(d+3)$-clique $H_i$. We pick an arbitrary edge in $H_i$ and subdivide it. We then identify $v_i$ with the vertex created by the subdivision (an alternative way of describing this is to delete an arbitrary edge in $H_i$ and then connect its previous endpoints to $v_i$).

Given an exact $(2,d)$-coloring of $G$, we extend it to an exact $(2,d+2)$-coloring of $G'$ by giving to every vertex in each $H_i$ the same color as $v_i$. Each vertex in each $H_i$ has exactly $d+2$ neighbors, all of its own color. Furthermore, each $v_i$ has the same neighbors as before (including $d$ of its own color) and two additional new neighbors of its color in its copy of $H_i$.

On the other hand, if we have an exact $(2,d+2)$-coloring of $G'$, due to the internal regularity of the $H_i$, they are all forced to be monochromatic and each $v_i$ must have the same color as its $H_i$. Then, by restricting the coloring to $G$, each $v_i$ loses exactly two neighbors of its own color, which means that the resulting coloring is an exact $(2,d)$-coloring.
\end{proof}

\begin{figure}[t]
\centering
\includegraphics[scale=1,keepaspectratio]{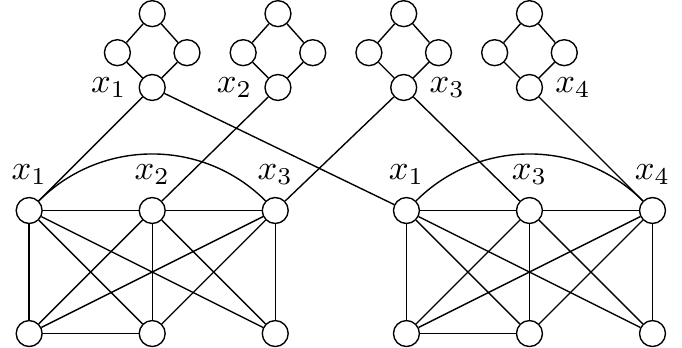} 
\caption{An instance $(x_1 \vee x_2 \vee x_3) \wedge (x_1 \vee x_3 \vee x_4)$ of \probMonotoneNaeThreeSAT transformed into a graph.}
\label{fig:ted-2k2d-reduction}
\end{figure}

Due to the nature of the inductive argument we used, we need a second base case to deal with even values of $d$. 

\begin{lemma}
\label{2ed2col}
\probTwoEDTwoCOL is $\NP$-complete.
\end{lemma}
\begin{proof}
We will show this by a polynomial-time reduction from \probMonotoneNaeThreeSAT. In this problem we are given a boolean formula in 3-CNF with only positive literals. The goal is to decide whether or not there is a satisfying assignment where each clause has at least one false literal (in addition to at least one true literal). This problem is $\NP$-complete by Schaefer~\cite{Schaefer1978}.
Our reduction is as follows: Let $H = K_3 + (K_2 \cup K_1)$, i.e., the join of the disjoint union of a vertex and an edge with a 3-clique, or equivalently a $6$-clique missing two incident edges. For each variable clause $C_j$, we create a copy $H_j$ of $H$, and label the vertices of the $K_3$ (in the join description) with the literals of $C_j$. Additionally, for each literal $x_i$ in the formula, we create a copy of $C_4$ and select one of its vertices to also label with $x_i$. We then connect the labeled vertices from each copy of $H$ to the labeled vertex from the corresponding copy of $C_4$ whose label they share. The construction is illustrated in Figure~\ref{fig:ted-2k2d-reduction}. We claim that the above graph admits an exact $(2, 2)$-coloring if and only if the given formula is satisfiable as stipulated.

Suppose that a satisfying assignment exists. For every literal which is true, we will color its corresponding copy of $C_4$ with color~$1$. Similarly, false literals will have their $C_4$ colored with color~$0$. Since these $C_4$ are internally $2$-regular, we will color each labeled vertex in a copy of $H$ with the opposite color of the corresponding $C_4$. At this point, each copy of $H$ will have two vertices of one color and one of the other, depending on how many true literals the corresponding clause has (which must be either one or two). Referring to the join definition of $H$, we can color the $K_1$ with the same color as two of the already colored vertices and the $K_2$ with the opposite color. This creates two monochromatic $3$-cycles in each copy of $H$, thus giving us the required coloring.

For the other direction, suppose that the given graph has an exact $(2, 2)$-coloring, using colors $0$ and $1$. Consider the $K_1$ in each copy of $H$. Since it has three neighbors (the labeled vertices in its copy), exactly two of those must share its color. Since these two vertices are also adjacent, they complete a monochromatic 3-cycle, so all their neighbors must be of the opposite color. This includes all three other vertices in the copy of $H$ in question, which also form a $3$-cycle. As such, every labeled vertex in a copy of $H$ must have the opposite color from the corresponding labeled vertex in a copy of $C_4$. Finally, the copies of $C_4$ must all be monochromatic. We assign a truth value to a given literal according to the color its copy of $C_4$ has (note that the choice of which color will correspond to ``true" and which to ``false" does not matter, as long as the assignment is consistent). Since there will be two labeled vertices of one color and one of the opposite color in each copy of $H$, each clause will have at least one true and at least one false literal, as required.
\end{proof}

We note here that copies of $C_3$ could have been used, instead of $C_4$. We decided to use $C_4$ to avoid confusion in the proof without increasing its length by making the construction's labels more explicit.

By combining the results of this section, we are now ready to give our main result on the complexity of exact defective coloring.

\begin{theorem}
\label{edcolfull}
\probtEDkCOL is $\NP$-complete for all $d \geq 1$ and $k \geq 2$.
\end{theorem}
\begin{proof}
This follows from Theorem~\ref{thm:1ed2col-npc} and lemmata~\ref{lem:col-to-ed},~\ref{2ed2col}, and~\ref{ed2col-induction}.
\end{proof}

We also observe that for $k = 1$, the problem \probtEDkCOL reduces to the question of whether a graph is $d$-regular, which is easily solvable in polynomial time. This establishes the computational complexity for all possible values of $d$ and $k$ for \probtEDkCOL.

\section{Algorithms for exact $d$-defective coloring structured graphs}
\label{sec:algos}
In this section, we give positive algorithmic results for finding exact $(k,d)$-colorings of certain structured graphs.
We begin by giving an efficient algorithm for deciding whether a given graph of bounded treewidth admits an exact $(k,d)$-coloring by exploiting known properties of such graphs.
While the metatheorems we apply are powerful as tools for classifying certain problems as admitting efficient solutions, the resulting algorithms tend not to reveal much about the combinatorial structure of these problems.
For cactus graphs, which have treewidth at most two and generalize cycles, we give a direct combinatorial algorithm that does not rely on any metatheorem.
Finally, we also give combinatorial algorithms for block graphs which, in turn, have bounded cliquewidth.

\subsection{Graphs of bounded treewidth}
Given that we have determined the exact $d$-defective chromatic number for e.g., trees and cycles, it is interesting to consider the question for more general sparse or ``tree-like'' graphs. Such graphs are captured by the notion of treewidth, measuring the distance of the given graph to a tree.
We begin with the following MSO$_1$ formulation of the problem of finding an exact $(k,d)$-coloring.

\begin{lemma}
\label{lem:mso}
For every $k,d \in \mathbb{N}$ there exists a \emph{MSO}$_1$ formula $\phi_{k,d}$ such that for every graph $G$, it holds that $G\models \phi_{k,d}$ iff $G$ is a YES-instance of \probtEDkCOL.
\end{lemma}
\begin{proof}
Our goal is to partition the vertices of $G=(V,E)$ into $k$ color classes $C_1, C_2, \ldots, C_k$ such that if a vertex $v$ is colored $i$ (i.e., $v \in C_i$), then there are exactly $d$ neighbors of $v$ that are also colored $i$. To achieve this, let us consider the MSO$_1$ formula
\begin{equation*}
\begin{split}
\psi_k 	&\coloneqq \exists C_1,\ldots,C_k \subseteq V \Big( \forall v \in V \Big( v \in C_1 \vee \cdots \vee v \in C_k \Big) \Big) \\
	&\wedge \Big( \forall i,j\in [k], i\neq j: (C_i \cap C_j = \emptyset) \Big) \\
	&\wedge \forall v \in V \Big( \bigwedge_{1 \leq i \leq k} ((v \in C_i) \implies \colnei(v,V,d,C_i) \wedge \neg \colnei(v,V,d+1,C_i)) \Big),
\end{split}
\end{equation*}
where the auxiliary predicate is defined as
\begin{equation*}
\begin{split}
\colnei(v,V,d,C) &\coloneqq \exists u_1,\ldots,u_d \subseteq V \Big( \forall j \in [d]: \adj(u_j, v) \wedge u_j \in C \Big).
\end{split}
\end{equation*}
Here, we require the desired partition into $k$ color classes to exist, and additionally stipulate that each vertex, when it belongs to color class $C_i$, has exactly $d$ vertices adjacent to it that are also in $C_i$.
The latter is established by the auxiliary predicate which expresses that $v$ which is in $C_i$ has $d$, but not $d+1$, neighbors that are also in $C_i$.
\end{proof}

\begin{lemma}
\label{lem:tw-poly}
Let $k,d,p \in \mathbb{N}$ be fixed. Then \probtEDkCOL can be solved in time $O(n)$ on $n$-vertex graphs of treewidth at most $p$. Furthermore, the problem can be solved in time $O(n^3)$ on $n$-vertex graphs of cliquewidth at most $p$.
\end{lemma}
\begin{proof}
The proof follows from Lemma~\ref{lem:mso} in conjunction with Theorem~\ref{thm:msotreewidth} and Theorem~\ref{fact:msorankwidth}.
\end{proof}
In other words, this result states that \probtEDkCOL is FPT parameterized by both the number of colors $k$ \textit{and} treewidth $p$. However, by Corollary~\ref{cor:explicit-ted-bounds}, it suffices to execute the algorithm of Lemma~\ref{lem:tw-poly} at most $p+1$ times. We arrive at the following stronger result.
\begin{theorem}
\probtEDkCOL is FPT parameterized by treewidth.
\end{theorem}

\subsection{Cactus graphs}

Since each cactus graph $G$ has a vertex of degree at most two, ${\chi_d^=(G) = \infty}$ for any $d > 2$. Thus, for any cactus $G$, it suffices to consider $\chi_d^=(G)$ for $d \in \{1,2\}$.

For completeness, we write out explicitly the result for $d=1$.
\begin{proposition}
A cactus graph $G$ has $\chi_1^=(G) \leq 3$ if and only if $G$ has a perfect matching. Moreover, this bound is tight.
\end{proposition}
\begin{proof}
The first claim follows directly from Corollary~\ref{cor:explicit-ted-bounds}.

To see that the bound is tight, consider the 6-vertex graph $H'$ obtained by taking a $K_3$ and adding a pendant to each of its vertices. The contraction of the (unique) perfect matching $M$ of $H'$ leaves a $K_3$ which is 3-chromatic. By taking a disjoint union of copies of $H'$ we have an infinite family of cactus graphs witnessing the obtained bound is tight.
\end{proof}

Let us proceed to give a polynomial-time algorithm to determine if $\chi_2^=(G) = 2$ for a cactus graph. To describe the algorithm, we will first introduce a definition we will need. For cliques, there already exists the notion of a \emph{simplicial} vertex, which is a vertex which belongs to exactly one maximal clique. We will now introduce the related notion of a \emph{cycle-simplicial} vertex. In a cactus graph, we will call a vertex cycle-simplicial if it belongs to exactly one cycle. Note that such a vertex can belong to multiple blocks, but only one of them can be a cycle, with the rest being cut-edges.

We are now ready to introduce Algorithm~1 for determining if a given cactus graph has an exact $(2,2)$-coloring. 
The reader should note that, to reduce clutter, we assume that the subroutines can access variables declared outside of their scope as well.
Similarly, we assume that the main algorithm can access variables (if any) declared inside the subroutines. 
In programming terms, these variables would constitute as being global. 

Let $G$ be a cactus graph on the vertices $u_1$, $u_2, \ldots, u_n$ with $r$ cycle blocks $V_1$, $V_2, \ldots, V_r$.
As a first step, Algorithm~1 preprocesses $G$ to build an auxiliary graph~$G'$.
Here, $G'$ is obtained by introducing a special vertex $x$, a vertex $v_i$ for each cycle $V_i$ where $i \in [r]$ and a vertex $w_i$ whenever the cycle $V_i$ contains a cycle-simplicial vertex.
For each $i$, if $w_i$ is added then the edges $w_iv_i$ and $w_ix$ are added.
Finally, whenever two distinct cycles $V_i$ and $V_j$ for $i,j \in [r]$ such that $i \neq j$ share a vertex, $v_i$ and $v_j$ are connected by an edge (see Figure~\ref{fig:cactus-example}).
The construction of $G'$ is detailed starting at Line~28.

\begin{figure}[t]
    \centering
    \begin{minipage}{0.49\textwidth}
        \centering
        \includegraphics[scale=1.0,keepaspectratio]{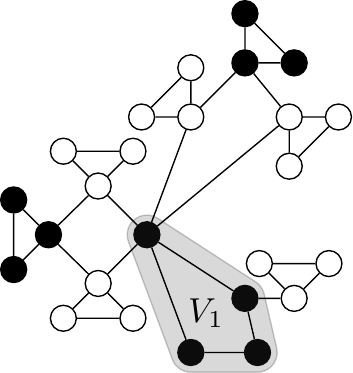} 
    \end{minipage}\hfill
    \begin{minipage}{0.49\textwidth}
        \centering
        \includegraphics[scale=1.0,keepaspectratio]{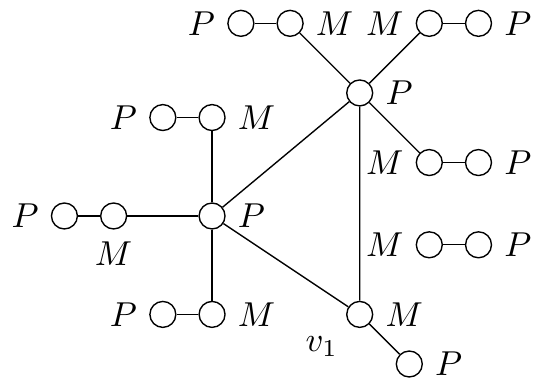} 
    \end{minipage}\hfill
    \caption{A cactus graph $G$ with an exact $(2,2)$-coloring (left) and a corresponding auxiliary graph $G'$ with an $\{M,P\}$-labeling (right). One cycle $V_1$ is highlighted in $G$ and its corresponding vertex $v_1$ labeled in $G'$. The vertices of degree one labeled $P$ are the vertices $w_i$.
    For brevity, other labels and the special vertex $x$ and its incident edges are omitted from $G'$.}
    \label{fig:cactus-example}
\end{figure}

The idea behind Algorithm~1 is that it marks the cycles of $G$ as either monochromatic (label $M$) or polychromatic (label $P$) under the $2$-coloring to be constructed. To do so, we exploit $G'$ whose vertices correspond to the cycles of $G$. 
The algorithm includes two subroutines (Line~44 and Line~51) which check local conditions for a valid coloring. Finally, once all cycles have been properly marked, the desired coloring can be constructed.

\begin{algorithm}
  \KwInput{A cactus graph $G$}
  \KwOutput{Decide whether or not $G$ has an exact $(2,2)$-coloring}
  
  \SetAlgoLined
  \DontPrintSemicolon
  \SetKwFunction{preprocess}{Preprocess}
  \SetKwFunction{algo}{}
  \SetKwFunction{procp}{LocalSubroutineP}
  \SetKwFunction{procm}{LocalSubroutineM}
  \SetKwFunction{preproc}{Preprocess}
  
  {
  
  Let $V(G) = \{u_1, u_2, \ldots, u_n\}$\;
  Let $\{V_1, V_2, \ldots, V_r\}$ be the cycles in $G$\;
  Initialize $G'$ to $\texttt{Preprocess}(G)$\;
  
  
  \For{$i = 1$ \KwTo $n$}{
  	\If{$U_i = \emptyset$}{\KwRet{NO}}
  }  

  Label $x$ with the label $M$\;
  
  \For{$i = 1$ \KwTo $r$}{
  	\If{$w_i \in V'$}{
  		Label $w_i$ with $P$\;
  		Label $v_i$ with $M$\;
  		\If{Calling \texttt{LocalSubroutineM}($v_i$) returns NO}{
  			\KwRet{NO}
  		}
  	}
  }   
  
  \Repeat{all vertices in $G'$ are labeled}{
  	\For{$j = 1$ \KwTo $n$}{
  		\If{$U_j$ contains some $v'$ labeled $M$ and some unlabeled $v_m$}{
			Label $v_m$ with $P$\; 
			\If{Calling \texttt{LocalSubroutineP}($v_m$) returns NO}{
  				\KwRet{NO}
  			}
  			\Continue
  		}
  		
  		\If{$U_j$ contains some unlabeled $v_m$ and all its other vertices are labeled $P$}{
			Label $v_m$ with $M$\; 
			\If{Calling \texttt{LocalSubroutineM}($v_m$) returns NO}{
  				\KwRet{NO}
  			}
  			\Continue
  		}
  	}
  }
  
  \KwRet{YES}
  
  }{}
    
  \caption{Algorithm for finding an exact $(2,2)$-coloring for a cactus graph}
\end{algorithm} 

\begin{algorithm}
\LinesNumbered
\setcounter{AlgoLine}{27}
  
  \SetAlgoLined
  \DontPrintSemicolon
  \SetKwFunction{preprocess}{Preprocess}
  \SetKwFunction{algo}{}
  \SetKwFunction{procp}{LocalSubroutineP}
  \SetKwFunction{procm}{LocalSubroutineM}
  \SetKwFunction{preproc}{Preprocess}
  
  \SetKwProg{myproc}{Procedure}{}{}
  
  \tcp{Supporting subroutines for Algorithm~1}  
  \BlankLine
  \BlankLine
  
  \myproc{\preproc{$G$}}{
  
  Let $G' = (V', E')$ where $V' = \{x\}$ and $E' = \emptyset$\;  
  
  \For{$i = 1$ \KwTo $r$}{
	Set $V' = V' \cup \{v_i\}$\;
	\If{$V_i$ has a cycle-simplicial vertex in $G$}{
		Set $V' = V' \cup \{w_i\}$ and $E' = E' \cup \{ w_iv_i, w_ix \}$
	}  
  }
  
  \For{$i = 1$ \KwTo $r-1$}{
  	\For{$j = i + 1$ \KwTo $r$}{
		\If{$V_i$ and $V_j$ share a vertex in $G$}{
			Set $E' = E' \cup \{ v_iv_j \}$  	
		}
  	}
  }
  
  \For{$i = 1$ \KwTo $n$}{
	Let $U_i = \emptyset$\;
	\For{$j = 1$ \KwTo $r$}{
		\If{$u_i \in V_j$}{
			$U_i = U_i \cup \{ v_j \}$		
		}	
	}  
  }
    
  \KwRet{$G'$}    
    
  }  
  
  \BlankLine
  \BlankLine  
  
  \SetKwProg{myproc}{Procedure}{}{}
  \myproc{\procp{$v_i$}}{
  \If{$V_i$ is a cycle of odd length in $G$}{\KwRet{NO}}
  \For{$j = 1$ \KwTo $n$}{
    \If{every vertex of $U_j$ is labeled $P$}{\KwRet{NO}}}
  \KwRet{MAYBE}\;}
  
  \BlankLine
  \BlankLine  
  
  \SetKwProg{myproc}{Procedure}{}{}
  \myproc{\procm{$v_i$}}{
  \If{any neighbor of $v_i$ is labeled $M$}{\KwRet{NO}}
  \KwRet{MAYBE}}
  
\end{algorithm} 

\begin{lemma}
\label{lem:extract-color}
Algorithm~1 correctly finds an exact $(2,2)$-coloring for a given cactus graph $G$ or decides that no such coloring exists.
\end{lemma}
\begin{proof}
The proof is broken into three parts. First, we give a high-level overview of what the solution (if any) must look like. Then, we observe that the algorithm halts, and finally prove that it produces (i) no false positives and (ii) no false negatives.

\paragraph{Intuition} Let $G'$ be the preprocessed auxiliary graph obtained from $G$ by the subroutine on Line~28.
By construction, the vertices $v_i$, for $i \in [r]$, of $G'$ correspond to the cycles of $G$. Similarly, the vertices of $G$ correspond to the maximal cliques of $G'$ (denoted by the sets $U_j$), excluding the edges that involve any $w_i$.
We also observe that, since any exact $(2,2)$-coloring will result in two $2$-regular color classes, these must necessarily be sets of disjoint cycles. In a cactus graph, all cycles correspond to blocks of the graph, which means that we must select two disjoint subsets of these cycle blocks to turn into color classes. 
As a first step after the construction of $G'$, we ensure (Lines~4--6) that each vertex of $G$ is on some cycle which is a necessary condition for the existence of a valid exact $(2,2)$-coloring.
Since each vertex must have exactly two neighbors of its own color, they must all belong to a specific cycle. As such, every vertex in $G$ has to belong to exactly one monochromatic cycle in a valid exact $(2,2)$-coloring; any other edge incident to it cannot be induced in a color class and must therefore not be monochromatic. The two subroutines $\procp{}$ and $\procm{}$ ensure that this is the case. Whenever a cycle is labeled $M$, $\procm{}$ checks that it does not share a vertex with another cycle already labeled $M$. Whenever a cycle is labeled $P$, $\procp{}$ checks that no vertex belongs only to cycles labeled $P$. This ensures that any vertex shared by more than one cycle will ``see" exactly one monochromatic cycle. As for cycle-simplicial vertices, these must force their cycles to be monochromatic, which is accomplished with the $w_i$ vertices.
An example of an $\{M,P\}$-labeling of an auxiliary graph $G'$ of a cactus graph $G$ is shown in Figure~\ref{fig:cactus-example}.

\paragraph{Halting}
Let us show that Algorithm~1 (Lines~1--27) halts, i.e., does not run indefinitely on a valid input cactus graph $G$.
It is obvious that the subroutines (Line~28, 44 and 51) halt.
Similarly, as Lines~1-13 obviously halt it suffices to show that the looping procedure from Line~14 halts, i.e., it indeed labels all vertices of $G'$ (unless it halts early, by returning a negative response). In particular, we will show that the loop always labels a vertex when traversed.

We will show this by contradiction. Note that every vertex $v_i$, corresponding to a cycle $V_i$ with a cycle-simplicial vertex, will be labeled $M$. Assume that the algorithm goes into an infinite loop and leaves some vertex, say $v_1$ of $G'$ unlabeled. Then, it must be the case that every clique that $v_1$ belongs to has some number of vertices labeled $P$ and at least one vertex that is also unlabeled, with no vertices labeled $M$. If this were not the case, then one of the if-statements (Line~16 and Line~21) within the loop would cause $v_1$ to be labeled. Since $v_1$ is not labeled, $V_1$ does not contain cycle-simplicial vertices. This implies that each of its vertices (which number at least three) must be shared with other cycles and for each such vertex there must be at least one corresponding unlabeled cycle (otherwise the loop would resolve $v_1$). Thus, there must exist at least three unlabeled neighbors of $v_1$. Select one such unlabeled neighbor, say $v_2$. As $v_2$ is unlabeled, $V_2$ must also have no cycle-simplicial vertices, so it must also have at least three unlabeled neighbors. Select one such neighbor, say $v_3$, making sure it is not $v_1$. By continuing in this manner, we can get an infinite sequence of unlabeled vertices in $G'$ without repetitions and so that no vertex in this sequence will be adjacent to any other vertices, except the ones immediately preceding or following it. This creates an infinite path in $G'$ which is impossible when $G$ is finite. As such, the algorithm above must label every vertex in $G'$ or end prematurely, which means it must halt. Note that this implies that every time the looping procedure of Line~14 is traversed, a vertex is labeled.

\paragraph{Correctness}
We will now give an exact $(2,2)$-coloring $c: V \rightarrow \{0, 1\}$ from the labels used in $G'$, i.e., after Algorithm~1 has successfully returned YES from Line~27. In each connected component of $G$, arbitrarily select a vertex and color it with $0$. For every uncolored neighbor~$v$ of a colored vertex~$u$ in $G$, we proceed as follows.

\begin{itemize}
\item If the edge $uv$ is in a cycle in $G$ whose corresponding vertex in~$G'$ is labeled $M$, we color~$v$, and ultimately the entire cycle, with $c(u)$.

\item If $uv$ is in a cycle in $G$ whose corresponding vertex in $G'$ is labeled $P$, we color $v$ with $1 - c(u)$, which causes the corresponding cycle to alternate between the two colors $0$ and $1$ (here we draw attention to Line~45 which ensures that this is possible to do).

\item Finally, if $uv$ is a cut-edge, color $v$ with $1 - c(u)$.

\end{itemize}
Observe that, since every vertex $u'$ in $G$ is in exactly one cycle labeled $M$, this will be a valid coloring, since all neighbors of $u'$ outside this cycle will not share its color. We also note that this coloring is well-defined, due to the structure of $G$. Given any two vertices in a connected component, every simple path between them will go through the exact same block sequence. We see, therefore, that the algorithm returns no false positive results.

Now, we will show that the algorithm returns no false negative results. Assume that some cactus graph $G$ has an exact $(2,2)$-coloring $c$. As discussed, every block of $G$ must be either a monochromatic cycle, a cycle of alternating colors, or a bichromatic edge. We claim that Algorithm~1, when applied to $G$, must complete and return YES. 
First, suppose that Algorithm~1 returns NO before before we reach the loop on Line~14. This happens either on Line~6 or Line~13.
\begin{itemize}

\item Line~6: Here, the algorithm stops if any clique $U_i$ is empty. This will happen if the vertex $u_i$ does not belong in any cycles in $G$. But if this is the case, $u_1$ cannot belong to a valid color class, a contradiction.

\item Line~13: Here, a graph is rejected if it has two cycles which have cycle-simplicial vertices and also share some vertex. Since both these cycles would need to be monochromatic, this creates a contradiction.

\end{itemize}
Therefore, if Algorithm~1 returns NO, this must happen in the repeat loop, and therefore as part of calling one of the local subroutines. 
This happens on Line~46, Line~49 or Line~53.
\begin{itemize}

\item Line~46: If this causes the algorithm to end, then some odd cycle was labeled with $P$ by the algorithm. Clearly, this cycle must be monochromatic in $c$. In a similar manner, if the algorithm stops at Line~49, there must be some clique, say $U$, all of whose vertices are labeled $P$. That is to say, there is some corresponding vertex $u$ in $G$, which is in no monochromatic cycle by the algorithm. Under $c$, this vertex must be contained in some monochromatic cycle. In both cases, we have some cycle, say $V_1$ that is monochromatic in $c$, yet labeled with $P$ by Algorithm~1. This implies that there must be some other cycle, say $V_2$ which shares some vertex, say $u_1$ with $V_1$, so that $v_2$ is labeled $M$ under Algorithm~1, as this is the only way for $v_1$ to be labeled $P$ by the repeat loop. But then, since $V_1$ and $V_2$ share $u_1$, $V_2$ must be polychromatic under $c$. As such, $V_2$ also cannot contain a cycle-simplicial vertex. But then, since $v_2$ is labeled $M$ by the algorithm, there must be some clique, say $U_2$, distinct from $U_1$, which contains $v_2$ and all of whose other vertices are labeled $P$ by the algorithm. This clique will correspond to the vertex $u_2$, which will be in $V_2$ but not $V_1$. But then, there must be some cycle $V_3$, which contains $u_2$, is contained in $U_2$, is labeled $P$ by the algorithm and is monochromatic under $c$. By continuing in this manner, we can backtrack along some subsequence of traversals of the repeat loop in Algorithm~1 in reverse, which create a path in $G'$ whose labels are exactly the opposite from the ones implied by $c$. But, given the way in which Algorithm~1 determines labels, this path must reach some cycle which contains a cycle-simplicial vertex, yet is polychromatic in $c$, which is a contradiction.

\item Line~53: In this case, we must have some vertex, say $u_1$, contained in two cycles, both of which are labeled $M$ by the algorithm. In $c$, at least one of these cycles must be polychromatic. Once again, we have a cycle which takes a different label by Algorithm~1 than is implied by $c$. By following similar steps to the previous case, we can once again reach a contradiction.
\end{itemize}
This concludes the proof.
\end{proof}

The above lemma reveals another interesting fact: The $\{M, P\}$-labeling of any cactus graph with an exact $(2, 2)$-coloring is unique. Furthermore, depending on our choice of starting vertices for each connected component, every possible valid coloring (up to isomorphism of the color classes) can be obtained. Both of these facts are also true for exact $(k, 2)$-colorings, for any $k$, as can be seen later in this section. 

\begin{lemma}
Algorithm~1 runs in time polynomial in the size of the input cactus graph $G$.
\end{lemma}
\begin{proof}
\label{lem:alg-cactus-polytime}
Clearly, the construction of the auxiliary graph $G'$ from Line~28 executes in polynomial time as each loop does a polynomial amount of work in either the number of vertices $n$ or the number of cycles $r$, where $r \leq \frac{n}{3}$.
Moreover, these $r$ cycles are easily found in time polynomial in the number of vertices $n$ as $G$ is a cactus graph.
Similarly, the two other subroutines are also carried out in polynomial time each time they are called. 
Each $U_j$ can contain no more than $r$ vertices and each $v_i$ can have at most $r$ neighbors. Therefore, the first two loops of Algorithm~1 also run in polynomial time.
Finally, as shown in the proof of the previous statement, each execution of the loop on Line~14 labels an unlabeled vertex of $G'$.
As such, the loop is executed at most $r$ times.
Each traversal of the loop must check up to $n$ cliques in $G'$ and for each clique up to $r$ vertices must be checked for their labels. Subsequently, one of the local subroutines is called (which, as already discussed, runs in polynomial time).
Therefore, the algorithm runs in polynomial time.

\end{proof}

We have shown that Algorithm~1 is both correct and runs in polynomial time, so we have proved the following.
\begin{theorem}
\probTwoEDTwoCOL is solvable in polynomial time on cactus graphs.
\end{theorem}
\noindent We can now modify the above result to obtain colorings with more than two colors. By also observing that any graph with an exact $(1, 2)$-coloring must be a disjoint union of cycles, we get the following result.

\begin{corollary}
\probTwoEDkCOL is solvable in polynomial time on cactus graphs.
\end{corollary}
\begin{proof}
Let $G$ be a cactus graph and $k > 2$. We can reuse Algorithm~1 with one simple modification: We remove Line~45 and Line~46 from $\procp{}$. This step was added to ensure that no cycle of odd length would be 2-colored. Since we are now working with more than 2 colors, that particular problem disappears. To obtain an exact $(k, 2)$-coloring, we slightly modify the process of Lemma~\ref{lem:extract-color}. When a cycle is labeled with $M$, we make it monochromatic, as before. If a cycle is labeled with $P$, its colors no longer need to alternate. We extend the coloring of the single vertex already colored in that cycle to an arbitrary proper $k$-coloring (i.e. without monochromatic edges). Finally, when dealing with a cut-edge, we must simply ensure that it is not monochromatic as well. Note that, since the coloring of each connected component of $G$ is done in a manner similar to a tree search, this cannot result in an ill-defined coloring.
\end{proof}

\subsection{Block graphs}
We now turn our attention to block graphs, which have bounded cliquewidth like cactus graphs do.

We will use the following problem and result as a subroutine.
In the \probCompleteCover problem, we are given a graph $G=(V,E)$ and asked whether $V$ can be partitioned into copies of $K_r$ for some fixed $r \geq 3$. The problem is $\NP$-complete in general~\cite{Garey1979}, but solvable in linear time on chordal graphs~\cite{Dahlhaus1998}.
\begin{theorem}[Dahlhaus and Karpinski~\cite{Dahlhaus1998}]
\label{thm:tr-cover}
For every $r \geq 3$, the \probCompleteCover problem can be solved in linear time for chordal graphs.
\end{theorem}
We begin by giving the following result on finding exact $(k,2)$-colorings of block graphs.
\begin{lemma}
\probTwoEDkCOL is solvable in linear time on block graphs.
\end{lemma}
\begin{proof}
By definition, each color class in an exact $(k,2)$-coloring induces a subgraph of vertex-disjoint cycles of length at least three.
In a block graph $G$, all vertices of any cycle are contained in a single block.
Further, any set $S$ of more than three vertices of a block induce a graph that is not a cycle because all vertices of $S$ are adjacent by the structure of $G$.
It follows that in an exact $(k,2)$-coloring of a block graph $G$, each color class induces a subgraph of vertex-disjoint triangles.

To find an exact $(k,2)$-coloring of $G$, we note that each block graph is chordal and thus we apply Theorem~\ref{thm:tr-cover} to find a triangle cover $\mathcal{T}$ of $G$. If $\mathcal{T}$ does not exist, we conclude that no exact $(k,2)$-coloring exists because there is at least one vertex that cannot be placed into a color class. Otherwise, when $\mathcal{T}$ exists, we construct the graph $H=(V,E)$, where $V$ has a vertex for each triangle in $\mathcal{T}$ and an edge between two vertices whenever the corresponding triangles, say $T_1$ and $T_2$, are adjacent, i.e., there exists $x \in T_1$ and $y \in T_2$ such that $x$ is adjacent to $y$ in $G$.
Equivalently, $H$ can be described via contractions of $G$, so $H$ is also a block graph.
It is well-known that the chromatic number of a chordal graph, and thus a block graph, can be computed in linear time. 
So if $\chi(H) \leq k$, we output YES and otherwise answer NO.
Furthermore, if desired, a corresponding exact $(k,2)$-coloring for $G$ is obtained in a straightforward manner from $H$.
All steps of the algorithm can be executed in linear time so our claim follows.
\end{proof}
\noindent By observing, more generally, that each color class must induce a subgraph of vertex-disjoint cliques in an exact $(k,d)$-coloring of a block graph, an argument similar to the previous lemma gives us the following result.
\begin{theorem}
For every $(k,d) \in \mathbb{N}^+ \times \mathbb{N}^+$, the problem \probtEDkCOL is solvable in linear time on block graphs.
\end{theorem}


\bibliographystyle{abbrv}
\bibliography{bibliography}
\end{document}